\documentclass[conference]{IEEEtran}

\ifCLASSINFOpdf
 \usepackage[pdftex]{graphicx}  
\else
\fi

\usepackage[cmex10]{amsmath}

\usepackage[tight,footnotesize]{subfigure}

\usepackage{amsfonts}

\newenvironment{proof}[1][Proof]{\begin{trivlist}
\item[\hskip \labelsep {\bfseries #1}]}{\end{trivlist}}
\newenvironment{claim}[1][Claim]{\begin{trivlist}
\item[\hskip \labelsep {\bfseries #1}]}{\end{trivlist}}

\renewcommand\footnotemark{}

\begin{document}
\title{Distributed Intrusion Detection of Byzantine Attacks in Wireless Networks with Random Linear Network Coding}

\author{\IEEEauthorblockN{Jen-Yeu Chen}
\IEEEauthorblockA{Department of Electrical Engineering\\
National Dong Hwa University\\
Email:jenyeu@mail.ndhu.edu.tw}
\and
\IEEEauthorblockN{Yi-Ying Tseng}
\IEEEauthorblockA{Department of Electrical Engineering\\
National Dong Hwa University\\
Email:m9723064@ems.ndhu.edu.tw}
}

\maketitle

\let\thefootnote\relax\footnote{*This paper is also published in the open access journal:International Journal of Distributed Sensor Network}
\begin{abstract}
Network coding is an elegant technique where, instead of simply relaying the packets of information they receive, the nodes of a network are allowed to combine \emph{several} packets together for transmission and this technique can be used to achieve the maximum possible information flow in a network and save the needed number of packet transmissions. Moreover, in an energy-constraint wireless network such as Wireless Sensor Network (a typical type of wireless ad hoc network), applying network coding to reduce the number of wireless transmissions can also prolong the life time of sensor nodes. Although applying network coding in a wireless sensor network is obviously beneficial, due to the operation that one transmitting information is actually combination of multiple other information, it is possible that an error propagation may occur in the network. This special characteristic also exposes network coding system to a wide range of error attacks, especially Byzantine attacks. When some adversary nodes generate error data in the network with network coding, those erroneous information will be mixed at intermeidate nodes and thus corrupt all the information reaching a destination. 
Recent research efforts have shown that network coding can be combined with classical error control codes and cryptography for secure communication or misbehavior detection. Nevertheless, when it comes to Byzantine attacks, these results have limited effect. In fact, unless we find out those adversary nodes and isolate them, network coding may perform much worse than pure routing in the presence of malicious nodes. In this paper, a distributed hierarchical algorithm based on random linear network coding is developed to detect, locate and isolate malicious nodes. To the best of our knowledge, this paper is the first one in the literature that proposes a distributed   intrusion detection and isolation scheme to effectively conquer Byzantine attacks for Random Linear Network Coding in a wireless network.

\emph{Index Terms}---Random Linear Network Coding, Byzantine attacks, intrusion detection, network coding, wireless sensor network, locating, watchdog.

\end{abstract}

\IEEEpeerreviewmaketitle

\section{Introduction}
\subsection{Network Coding}
Network coding has become a paradigm shift in information transmission, it is first brought up by Prof. Shuo-Yen Robert Li et al \cite{RobertLi}. Instead of traditional information transmission method, storing and forwarding, network coding allows intermediate nodes to mix received information together and transmit new information generated by the received information in terms of encoding. Due to encoding operation at intermediate nodes, data can be regarded as information flow through network, which is in a sense of data compression. Therefore throughput and bandwidth efficiency can be increased and delay can be decreased also via network coding. In \cite{RobertLi}, Prof. Shou-Yen Robert Li has showed that network capacity with network coding can be bounded by min-cut max-flow theory, which is larger than traditional storing-and-forwarding method.

\subsection{Random Linear Network Coding}
Recent research's having proven throughput gain of network coding in variety of application makes network coding an attractive topic. With algebraic approaches, such as \cite{algebraicNC}, a communication pattern with network coding of a network can be designed and achieve its promised capacity, which is the min-cut from the source to the sinks in a network graph \cite{RobertLi}. However, algebraic approaches require much central information and optimized coding scheme is actually not practical to design at most time \cite{insufficiency}. Then a distributed method of network coding has been developed – Random Linear Network Coding, shorted as RLNC \cite{RLNC}. RLNC is a powerful tool to disseminate information in networks for it is distributed and robust against dynamic topology. Without knowing central information such as network topology, RLNC regards every encoded packet as a coding vector over a finite field $\Bbb{F}_q$ and generates new packets at intermediate nodes by linearly combining received packets with random coefficient. Some overhead in packet's header is introduced to record how packets are combined (in \cite{RLNC}, it is called global encoding vector) and sinks can do decoding and recover original information as long as they retrieve enough packets.

\subsection{Security issue of network coding}
Network coding shows its variety of possibilities and benefit in information dissemination, however, it also introduces new type of security issue. The most serious security challenges posed by network coding thus seem to come from various types of Byzantine attacks, especially packet-modifying attack. In particular, RLNC has been shown very robust to packet losses induced by node misbehavior \cite{RLNCrobusttoloss}. Nevertheless, when it comes to packet-modifying attack, RLNC has become quite vulnerable. In RLNC, one intermediate nodes will linearly combine received packets and generate new packets to next multiple receivers. If this node has been compromised and generates error packets, other nodes received those error packets will also be modified for those error packets will stay in buffer and keep being combined with normal packets. Hence, nodes of each path these error packets go through would become new compromised nodes without self-awareness and disseminate more error packets. In other word, the error due to modified packets will \emph{propagate} in network with RLNC. Eventually, the whole communication network may be crushed just because of one single adversary node. Fig.~\ref{fig_error_prop} shows how a single adversary node propagates error.

The paper is organized as follows: Section II illustrates pros and cons of related works on Byzantine attacks, Section III describes our model and algorithm, Section IV gives the simulation results and analysis, Section V shows mathematical analysis. Section VI concludes the paper with a summary of the results and discussion of further work.

\section{Related work}
Existing method mostly modifies the format of coded packet against Byzantine attacks, and can be divided into two main categories: (1) misbehavior detection, and (2) end-to-end error correction.
\subsection{Misbehavior Detection}
Misbehavior detection applies error control technique or information-theoretic frameworks of encryptography to detect the modification introduced by Byzantine attackers. By types of nodes who take care of coding burden, misbehavior detection can be further divided into \emph{generation-based} and \emph{packet-based}. \emph{Generation-based} detection takes similar advantage as error-correcting codes and lays expensive computation tasks on destination nodes. As long as enough information is retrieved by destinations, modification can be detected. \cite{byzantinemodification} proposes an information-theoretic approach for detecting Byzantine modification in networks employing RLNC. Each exogenous source packet is augmented with a flexible number of hash symbols that are obtained as a polynomial function of the data symbol. This approach depends only on the adversary not knowing the random coefficient of all other packets received by the sink nodes when designing its adversarial packets. The hash schemes can be used without the need of secret key distribution but the use of block code forces an priori decision on the coding rate. Moreover, the main disadvantage of generation-based detection schemes is that only nodes with enough packets from a generation are able to detect modifications and thus, result in large end-to-end delays.

\begin{figure}[!t]
\centering
\includegraphics[width=8cm]{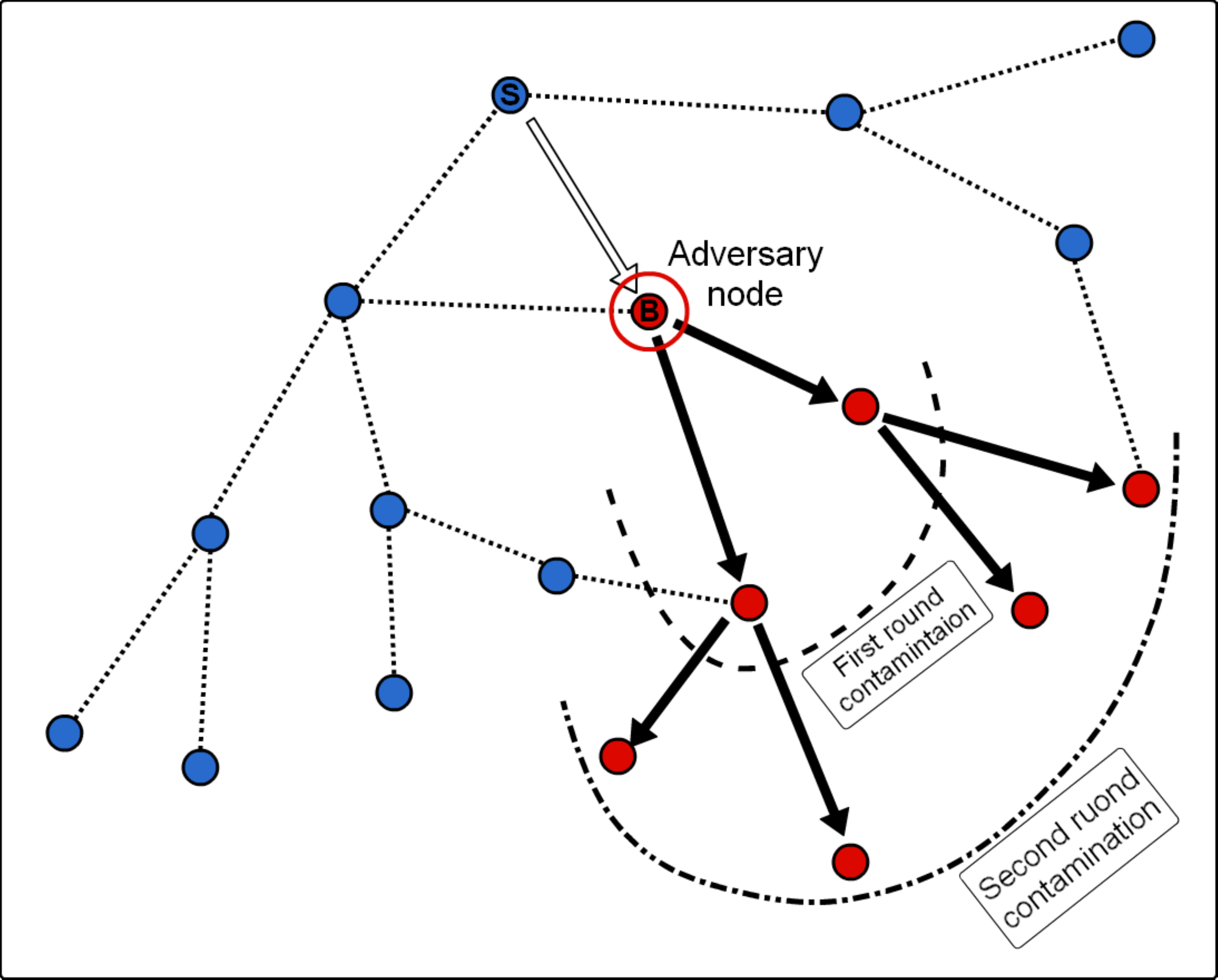}
\caption{Error propagation due to modifying packets by Byzantine nodes in a network with RLNC}
\label{fig_error_prop}
\end{figure}
On the contrary to generation-based detection schemes, \emph{packet-based} detection schemes allow intermediate nodes in the network detecting modified data on the fly and drop modified packets instead of only relying on destinations, which is more suitable for high attack probability compared to generation-based detection schemes. Packet-based detection schemes require active participation of intermediate nodes with ability to compute hash function or generate signature based on homomorphic hash functions \cite{signatureNC}, \cite{on-the-fly}. Hash of a coded packet can be easily derived from the hashes of previously encoded packets; in that way, intermediate nodes can verify validity of encoded packets before linearly combining them. This characteristic also prevents from error propagating in network. Unfortunately, homomorphic hash function is also computationally expensive and can't be used in inter-session network coding scenario while different sources combine their own source information together.

\subsection{End-to-end Error Correction}
End-to-end error correction schemes include error correcting code method into the process of encoding packets and sinks can correct error and recover original information under certain amount of error. Like generation-based detection schemes, end-to-end error correction schemes lay all encoding and decoding tasks on sources and sinks, such that intermediate nodes are not required to changer their mode of operation. The transmission mode for end-to-end error correction schemes with network coding can be described by matrix channel {\boldmath$Y = AX + Z$}, where {\boldmath$X$} is the matrix whose rows are the source packets, {\boldmath$Y$} corresponds to the matrix whose rows are received packets at sinks, {\boldmath$A$} denotes the transfer matrix, which records linear transformation operated on packets while they traverse the network, also called global encoding vectors, and {\boldmath$Z$} describes the matrix according to the injected error packets after propagate over the network. With error-correcting code, we can recover {\boldmath$X$} from {\boldmath$Y$}. \cite{codingforerr}, \cite{rank-metric} and \cite{usingrank-metric} discuss performance of error correction ability while some channel information, such as loss rate or error probability, is known. \cite{capacityforRNC} proposes a simple coding schemes with polynomial complexity for a probabilistic error model of random network coding and provides bounds on capacity. \cite{resilient} provides a special coding method, which adds a zero vector in the transmitted packet at the source node with assumption that there is a secret channel between source nodes and sink nodes to inform sinks where the zero vector locates in the transmitted packet. This information can't be seen by intermediate nodes and it will be very useful while Byzantine attackers maliciously modify the transmitted packet. As a matter of fact, under some modification level, the more modification, the more likely sinks can recover the original information by using information from observing modified zero vectors. \cite{resilient} also gives bounds on capacity for two adversarial mode: when Byzantine attackers have limited eavesdropping ability, optimal rate would be C-z; when Byzantine attackers can eavesdrop all links, optimal rate would be down to $C-2z$, where $C$ is the network capacity and $z$ is the number of links controlled by attackers. With special error-correcting code, sinks can be more tolerant with errors, but this scheme also introduces large overhead in packets which result in tremendous transmission efficiency decreasing.

Even though end-to-end error correcting schemes can recover original information at sinks, it can't stop error from propagating and introduces large overhead (in worst case, only ${1\over 3}$ of a packet carries data); misbehavior detection schemes can intercept modified packets on the fly to prevent errors from propagating, but it unfortunately takes expensive computation complexity. We will propose a new type of network coding packet and a distributed algorithm to locate Byzantine attackers and then isolate those nodes. Our algorithm essentially control the error propagation over the network and is not computationally expensive. Detailed introduction is in the next section.

\section{Network Model and Byzantine Attackers}
\subsection{Network Model with RLNC}
Consider a wireless network of $n$ nodes with communication range of $r$ randomly distributed in a square area, represented by an undirected graph $G=(V,E)$, with $\mid V\mid =n$ nodes. Let $d(i,j)$ denotes the distance from node $i$ to node $j$. An edge $e_ij \in E$ when $d(i,j)\leq r$. Besides, these $n$ nodes have the ability to access the information of their position. Without loss of generality, we assume the lower left corner of the square area to be the origin and each nodes know their coordinate such as $(3, 4)$.

In the communication pattern in which we are interested, each node can perform RLNC to disseminate messages. One source $S$ trying to multicast $k$ messages $\{m_1,\dots ,m_k\}$ to $d$ destinations $\{D_1,\dots ,D_d\}$ transmits those messages as vectors of bits which are of equal length $u$, represented as elements in the finite field $\Bbb{F}_q$, where $q=2^u$.The length of the vectors is equal in all transmissions and all links are assumed to be synchronized with a global clock splitting time into slots or rounds which are common to all nodes tin the network. In each time slot, nodes with messages in buffer send out new messages on edges to other nodes simultaneously. Let $S_i (t)=\{f_1,\dots ,f_{|S_i (t)|}\}$ be the set of all messages at nodes $i$ at time slot $t$, and by definition, for $f_l \in S_i (t) , 1\leq l\leq |S_i (t)| , f_l \in \Bbb{F}_q$ and $f_l = \sum_{u=1}^n{\alpha_{l_u}m_u}, \alpha_{l_u} \in \Bbb{F}_q$. When a node $i$ sends out a message , this message is actually a liner combination, called \emph{local encoding}, of the messages stored in node $i$ with payload $g_{i,out} \in \Bbb{F}_q$, where
\[
g_{i,out} = \sum_{f_l \in S_i (t)}{\beta_l f_l}, \beta_l \in \Bbb{F}_q; Pr(\beta_l = \beta) = {1\over q}, \forall \beta \in \Bbb{F}_q
\]
The vector {\boldmath$\beta$} $=[\beta_1,\dots ,\beta_{|S_i (t)|}]$ is called \emph{local encoding vector}, and the message $g_{i,out}$ can be further written as follows.
\begin{equation*}
\begin{aligned}
g_{i,out}=\sum_{f_l \in S_i (t)}{\beta_l f_l}= \sum_{f_l \in S_i (t)}{\beta_l  \sum_{u=1}^k{\alpha_{l_u}m_u}}\\
= \sum_{u=1}^k{\left(\sum_{l=1}^{|S_i (t)|}{\beta_l \alpha_{l_u}}\right)m_u}=\sum_{u=1}^k{\gamma_u m_u}, 
\end{aligned}
\end{equation*}
where $\gamma_u =\sum_{l=1}^{|S_i (t)|}{\beta_l \alpha_{l_u}} \in \Bbb{F}_q$ and the vector {\boldmath$\gamma$} $=[\gamma_1, \dots ,\gamma_k]$ is called \emph{global encoding vector}. The global encoding vectors are transmitter over the network for decoding and we define our transmitted packets as Fig. ~\ref{fig_packet_form} to assure that coefficients $\gamma_u$ are recoded and nodes know that.

\subsection{Threat Model and Our Algorithm}
We propose an algorithm, Distributed Hierarchical Adversary Identification and Quarantine, to fight against packet-modifying attack introduced by compromised Byzantine nodes. Assume $z_0$ out of $n$ nodes has been compromised as Byzantine nodes and they will modify every packet they send out in order to crash the whole network transmission. Specifically speaking, these Byzantine nodes modify the global encoding coefficients or payload of newly generated outgoing messages, which result in error due to that the modified vectors may not belong to the vector space spanned by source messages and further propagate the errors by following linear combinations of other nodes. We seek an algorithm to locate these Byzantine nodes and isolate them, so that they cannot affect the network.

As mentioned above, network coding is susceptible to the packet-modifying attacks for errors will propagate by operation of linear combinations. However, our algorithm, DHAIQ, uses this characteristic to let error propagate within a certain range in order to let some chosen nodes, referred as \emph{watchdogs}, detect that there are some Byzantine nodes in the monitored area. Before starting our algorithm, we assume that node density and is known by every nodes from operating other algorithm such as aggregate computation. DHAIQ can mainly divided into $5$ steps:\\
\begin{enumerate}
\item \label{st1}When a network is under packet-modifying attacks, an arbitrary node in the network will trigger the whole algorithm. This node is the watchdog of the $1st$ level. This first watchdog will awake the $2nd$ level's \emph{four} watchdogs and pass two messages, which are node density and the monitoring area size. The node density is a criterion of termination scheme and the whole deployment area is the $2nd$ level's monitoring range as figure ~\ref{sfig:step1} illustrates. The awaken watchdogs are chosen by locations. These four watchdogs are situated in each corner of their common monitoring area. After awaking the $2nd$ level's watchdogs, the first watchdog ends its monitoring mode and turns back to its normal mode.\\
\item \label{st2}Each of the $2nd$ level's watchdogs will generate its own special packet, referred as \emph{probe packet}. It then sends this probe packet to the other three watchdogs in an area-restricted flooding way as described in figure ~\ref{sfig:step2}. Except for these watchdogs, every node that receives these packets will do encoding and then sends new packets to all its neighbors. These packets will be linearly combined via intermediate nodes and constrained to disseminate within the monitoring range. This is all determined at the $2nd$ level. There are four watchdogs and obviously four different probe packets which are in the same \emph{generation}. The packets belonging to the same generation will start and terminate transmitting simultaneously based on a \emph{time stamp}. Any node that receives the probe packets the first time will record this time stamp. Nodes will continue encoding and sending out packets until the time stamp is expired. If a probe packet reaches a node outside the monitoring range, this node will drop that packet. The information carried by probe packets only traverse in the monitoring range. With the time stamp, all nodes that belong to the same monitoring area can terminate transmitting simultaneously. Before the termination of monitoring, all watchdogs keep retrieving packets from other nodes and keep a \emph{packet pool} in their buffer. An arriving packet is called \emph{innovative packet} only if it is linear independent to each packets stored in a watchdog's buffer. The discard rule is to keep innovative packets and drop all non-innovative packets. In this way, we also can limit buffer size to a pretty small value. There will be only four packets if there's no adversary node in the monitoring area. Watchdogs also keep computing the rank of vector space spanned by buffered packets until this generation is expired. \\
\begin{figure}[!t]
\centering
\includegraphics[width=8cm]{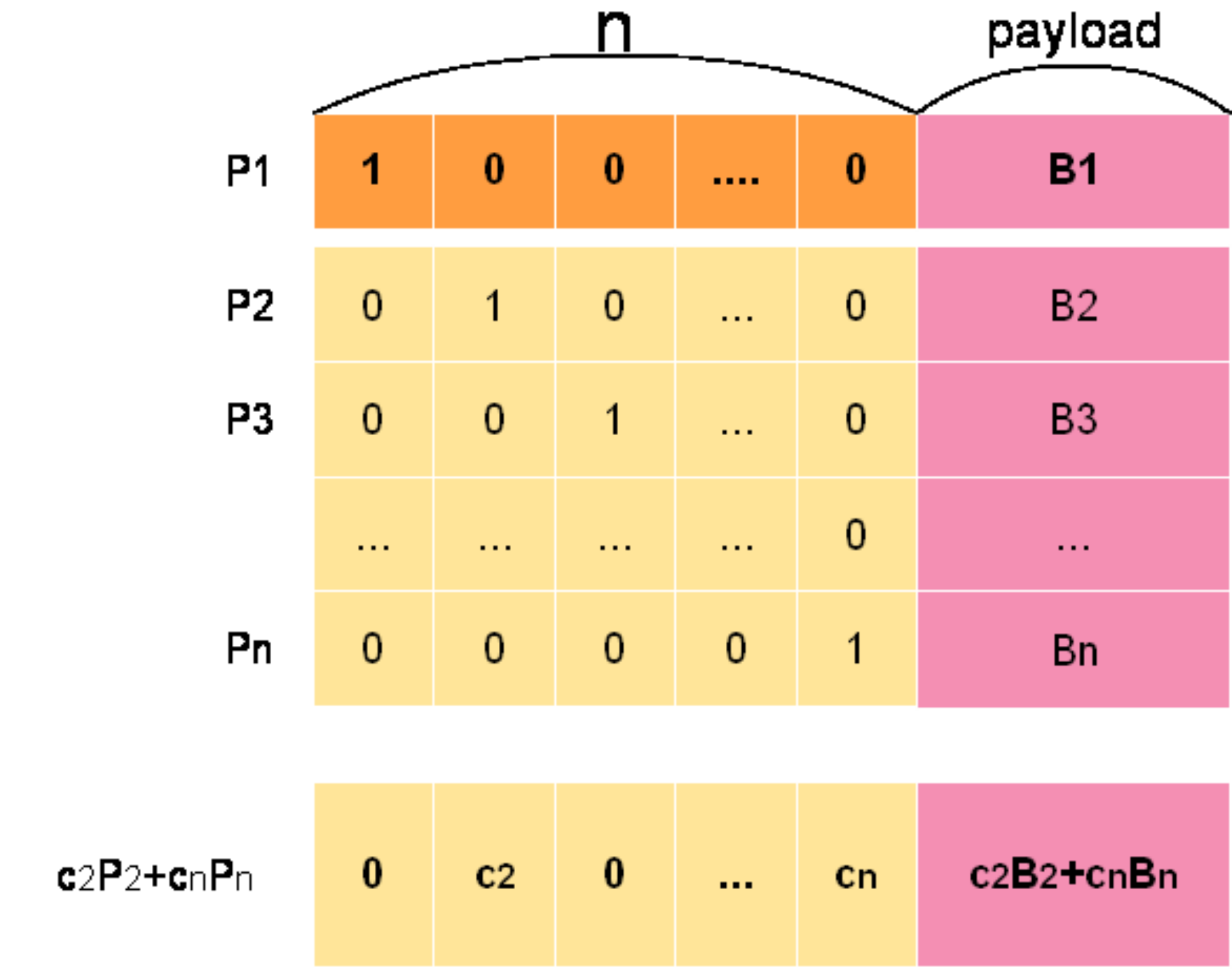}
\caption{The practical format of transmitted packets}
\label{fig_packet_form}
\end{figure}
\item \label{st3}If there is any adversary node in the monitoring areas, errors would propagate in the monitoring area and some of the watchdogs would receive modified packets with high probability. Watchdogs can judge whether they receive modified packets by the rank of packet pools. For example, one can say that there is at least an adversary node located in the monitoring area when a watchdog has a packet pool of rank $5$. As soon as any of watchdogs detects the existence of adversary nodes, that watchdog will notify the other watchdogs in the same generation and trigger the next level's watchdogs together as shown in figure ~\ref{sfig:step3}. These four watchdogs will divide their common monitoring rang into four sub-areas by their corners discussed previously. Each watchdog can then duplicate what the first watchdog does in step~\ref{st1}).  Each of them awakes four arbitrary nodes in its corresponding sub-area and pass node density and next level's monitoring range, which is a quarter of a current monitoring range according to the location of the upper level's watchdog. The awaken four nodes will also approximately locate at each corner of the sub-area and there will be a total of $sixteen$ watchdogs awaken for four sub-areas of the next level ($3_{rd}$ level) as displayed in figure ~\ref{sfig:step4}. \\
\item Repeat step~\ref{st2}) and step~\ref{st3}), keep dividing the areas in a distributed way until we can locate adversary nodes in a small enough area. We define this "small enough area" by the number of nodes locating in it. When the number is small and under a threshold $\lambda$, we terminate the monitoring of this area. The number of the node in an area can be estimated by the information of node density and monitoring range which are carried by probe packets. Therefore this "small enough area" will be the least monitoring area we can divide. In the least monitoring area, it is very possible that an adversary node is chosen as a watchdog. In this case, adversary nodes may realize this is the time to temporarily act normal and stop modifying the contents of packets. The detection will fail due to adversary nodes' temporary good behaviors. Any detection in progress will be terminated if its monitoring range is under the threshold and all the nodes in this area will be marked as suspect nodes.\\
\item After some random time intervals, another arbitrary node will trigger the algorithm again and this time its monitoring range will be shifted by a short distance. In the very end of the algorithm, we will mark some small squares which contain adversary nodes. If we shift the monitoring range a little in the beginning of the algorithm, the squares we choose will not be identically overlapped but partially overlapped. This partially overlapped area may contain adversary nodes with high probability and the other non-overlapped areas, which may contain normal nodes but remarked as suspect,  would be less suspicious. In this way, we can eliminate the number of nodes who are marked as suspects but in fact are normal nodes, referred as \emph{innocent nodes}. To get the final result, each node in the network maintains a suspect table. Whenever a node is reported as a suspect, its suspect level in the other nodes' tables increases by 1. The nodes with high suspect level will be regarded as adversary nodes and isolated. Our simulation results show this shift scheme can greatly reduce the amount of mistaken nodes.
\end{enumerate}
\begin{figure}[!ht]
\centering
\subfigure[The first watchdog awakes watchdogs of next level.]{
\includegraphics[width=3.5cm,keepaspectratio]{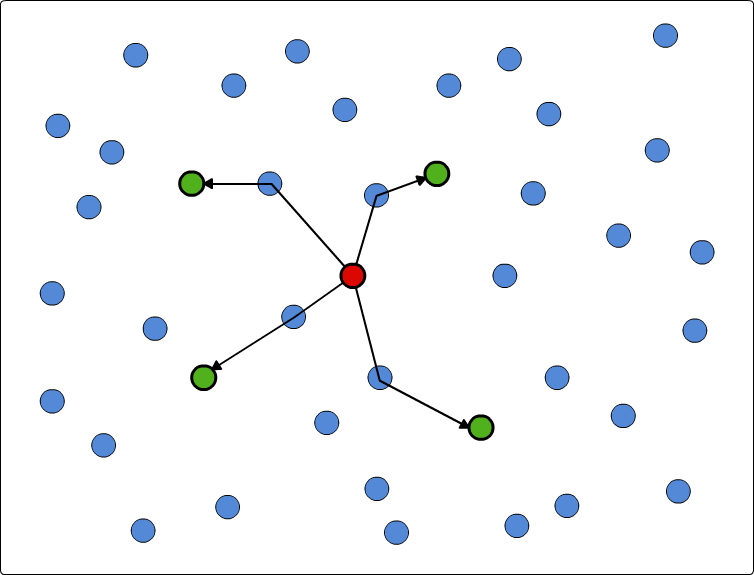}
\label{sfig:step1}
}
\subfigure[Watchdogs of next level start sending out probe packets.]{
\includegraphics[width=3.5cm,keepaspectratio]{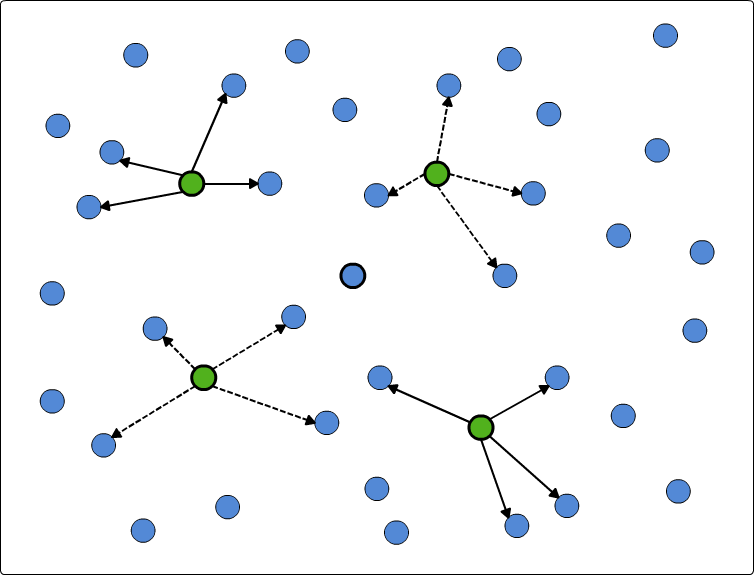}
\label{sfig:step2}
}
\subfigure[One watchdog detects errors and notifies the others.]{
\includegraphics[width=3.5cm,keepaspectratio]{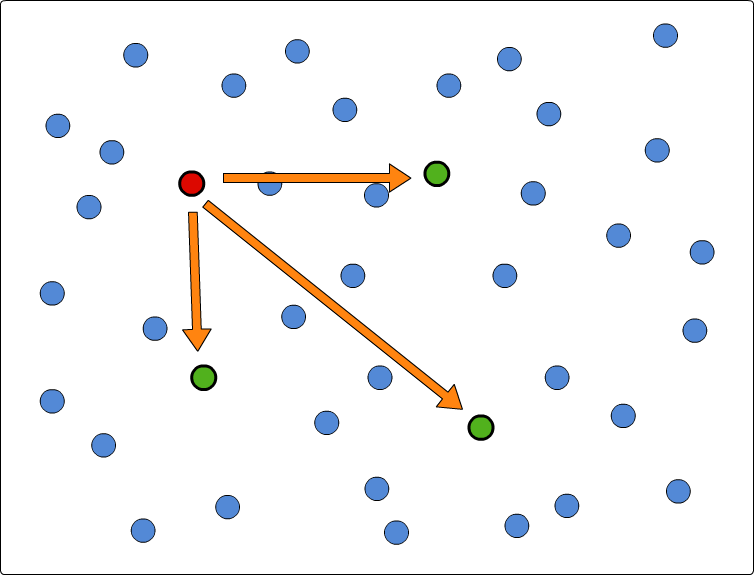}
\label{sfig:step3}
}
\subfigure[Each watchdog further awakes more watchdogs of next level.]{
\includegraphics[width=3.5cm,keepaspectratio]{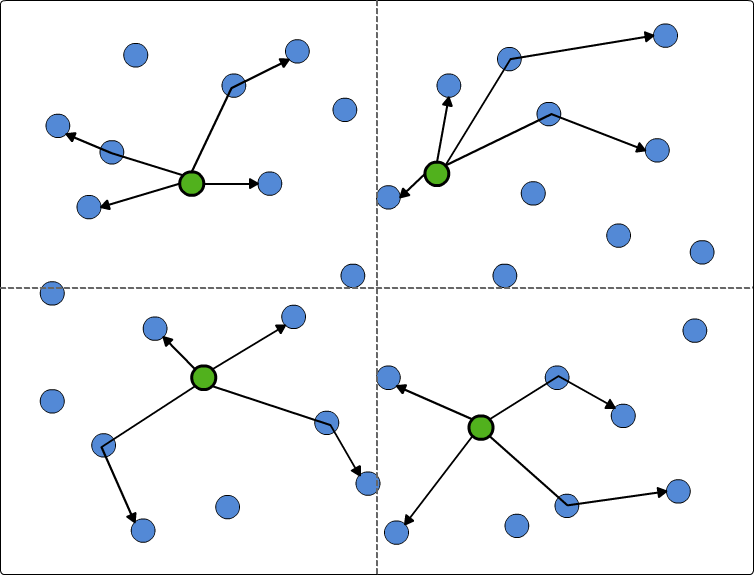}
\label{sfig:step4}
}
\caption{\label{fig:algorithm_step}Hierarchical division of the monitoring areas.}
\end{figure}
\section{Analysis and Simulation Result}
\subsection{Probe packets and time stamp} 
In most scenarios of RLNC application, the destinations do the decoding as long as they receive full rank of packets. In our algorithm, we modify this scheme that destinations don't decode to fit our requirements. Considering the worst case, to detect an adversary node is that all watchdogs gather around the center of the monitoring area and the adversary node is located at the very edge. Based on the flooding method, the least time slot required for watchdogs to receive modified packets is the hop number of the shortest path from the adversary nodes to the watchdogs, which is half diagonal of the monitoring area. Since the source of modified packets also come from watchdogs, the average number of hop for a modified packet to arrive the watchdogs is $\sqrt{2k}$. Note that $k$ is the node number of current monitoring area, which is accessible information for watchdogs. We can set time stamps of each generation with this number $\sqrt{2k}$ to assure that watchdogs can receive modified packets and trigger the next level whenever there are Byzantine nodes. When a time stamp is expired, its corresponding nodes will terminate disseminating packets and empty their buffer.
\subsection{Range of shifting}
Simply repeating the algorithm won't perform better since the sub-areas are equally divided. If the algorithm starts with the same monitoring area, it will eventually lead to the same result and be in vain. Thus we shift the starting monitoring area in order to minimize the number of innocent nodes. Now the question is how many we should shift each time. It is straightforward to see that if we shift more than a single least monitoring area, this shift is useless. Hence we know the shift range should be no larger than the length of edge of the least monitoring area.

The purpose that we use shift scheme is to further divide the least monitoring area into smaller areas so that we can eliminate the number of innocent nodes. To this end, we shift in both horizontal and vertical directions to let overlapped areas divide the least monitoring area into \emph{four} smaller areas. Hence the question has become how to divide these four smaller areas in order to get the least innocent nodes. Basically we have two options here, equal division and non-equal division. In fact, the equal division method will have the least expected value of innocent nodes. The mathematical analysis is in section V, and the simulation results also support our idea.
\subsection{Innocent nodes and overhead}
When we mark the nodes in the least monitoring area as suspect nodes, we mark all the nodes in the area. In fact, some nodes are normal nodes but marked as suspect, and we call them \emph{innocent nodes}.  Consider the case which we only perform identification algorithm once without using suspect table. It is straightforward that uniform distribution of Byzantine nodes can lead to the worst result with the most innocent nodes. The ratio of innocent nodes is upper bounded by $\dfrac{(\mu -1)z_0}{n}$ and this bound grows linearly with respect to the number of Byzantine nodes and $\mu$, which is quite a large number. Besides, probe packets carry no data information and the amount of probe packets transmitted of all generations in each level is $O(n\sqrt{n})$. In one identification algorithm, it will trigger $O(\log n)$ levels totally and therefore total number of transmitted probe packets is $O(n \sqrt{n}\log n)$ in time $O(\sqrt n)$.

\subsection{Simulation Results}
In our simulation, we uniformly distribute 400, 600 ,800 and 1000 nodes in a square area with width of 800 and node communication range is 50. We simulate our algorithm under the circumstance of the amount of adversary nodes varying from 5 to 45 and these adversaries are uniformly and normally distributed. Fig. ~\ref{result_origin} is the first result of our algorithm, we can see that the innocent ratio of uniform distribution pattern is quite high. The uniform distribution pattern is the worst case to our algorithm. In order to decrease the amount of innocent nodes, we introduce shift scheme. The result are shown in Fig. ~\ref{result_shift}. The result with more nodes is in Fig. ~\ref{result_more}. As we can see, our algorithm performs better in a dense topology. Performing shift scheme in our algorithm can eliminate innocent ratio effectively, but it also drags down the catch ratio a little bit. Because shift scheme also generates holes around boundaries, which can not be detected sometimes. The result shows that the catch ratio only drops a little, which is an acceptable value.

\begin{figure}[!t]
\centering
\subfigure[Innocent ratio of uniform distribution]{
\includegraphics[width=4cm]{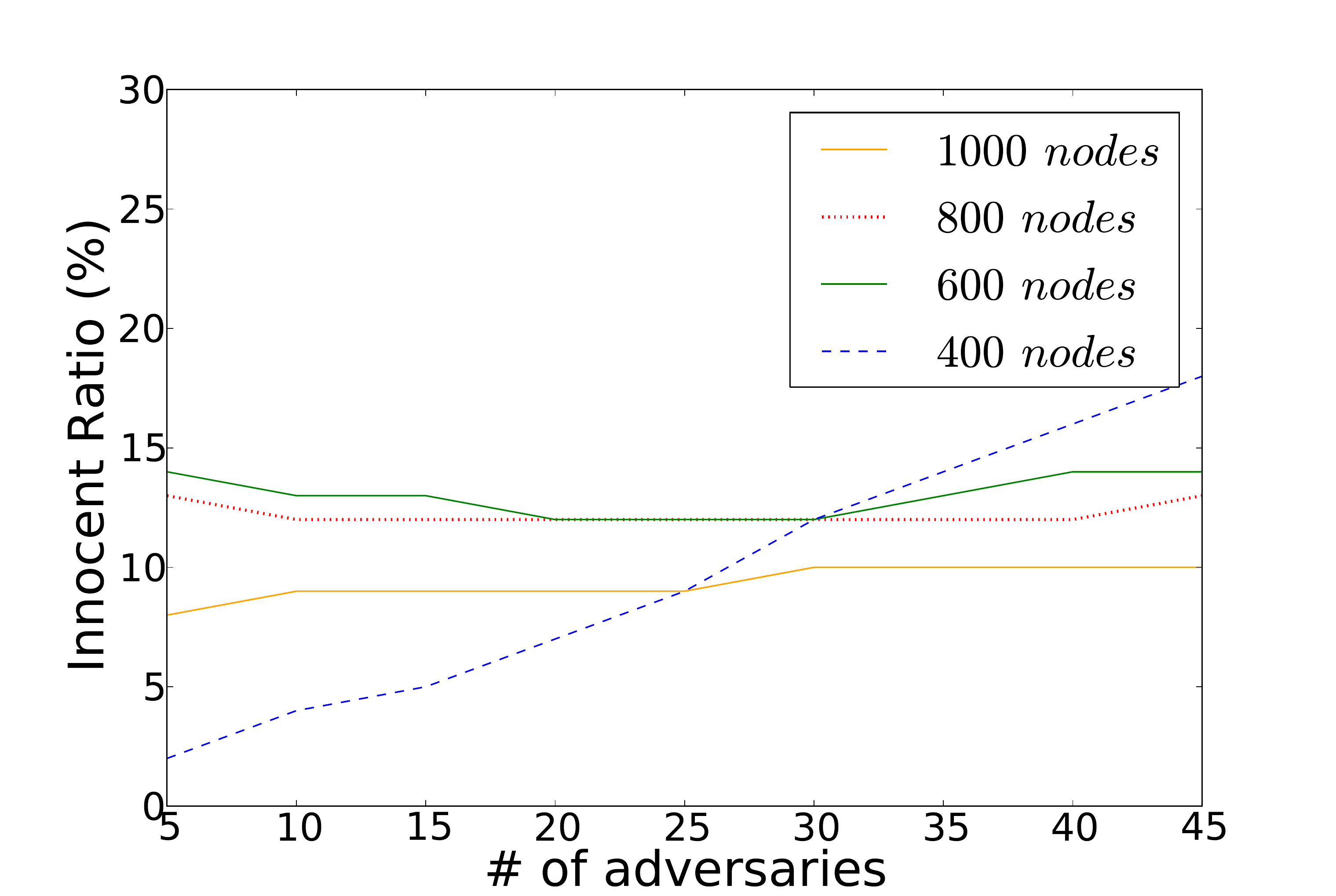}
\label{sf:in_uni}
}
\subfigure[Catch ratio of uniform distribution]{
\includegraphics[width=4cm]{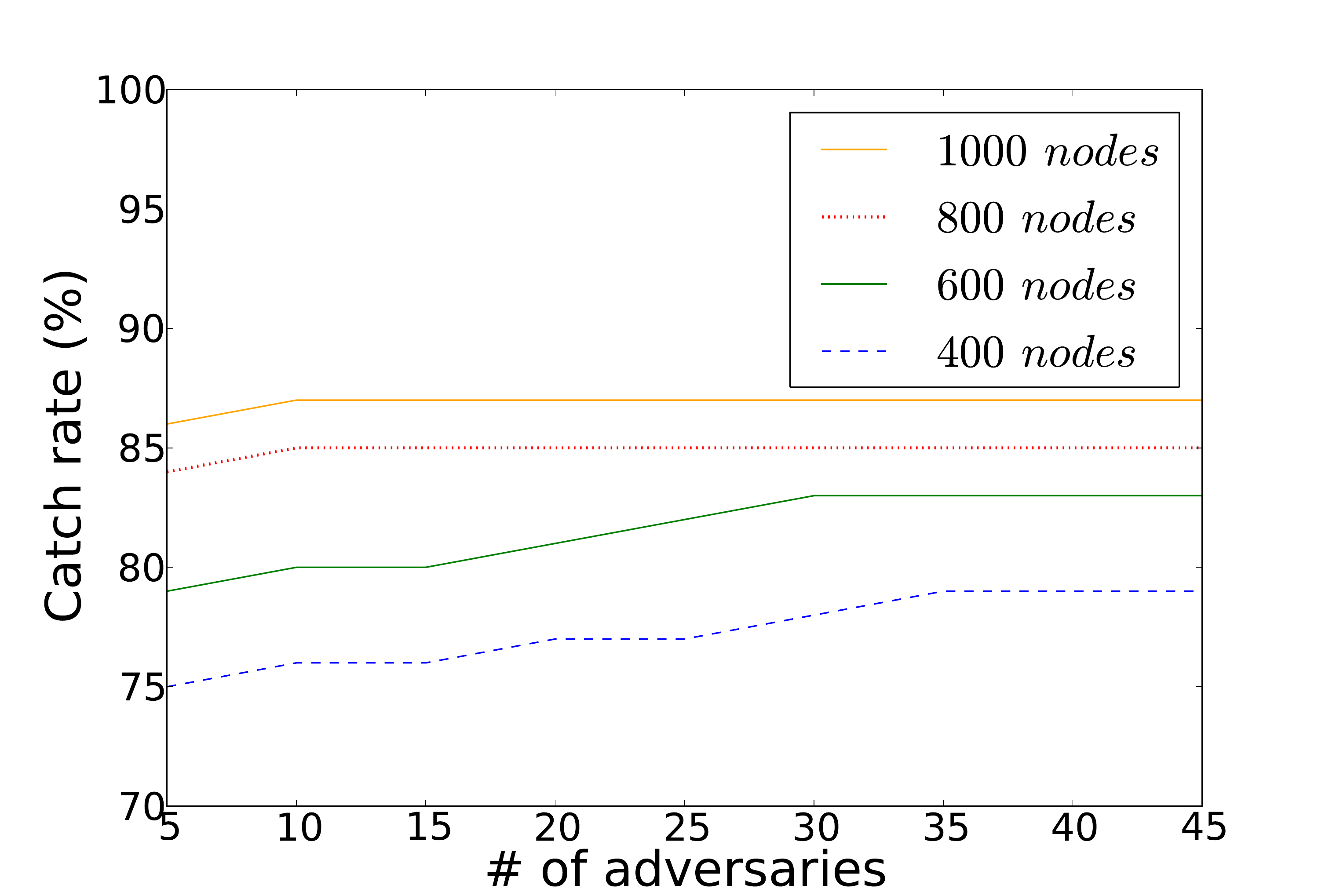}
\label{sf:cr_uni}
}
\subfigure[Innocent ratio of Gaussian distribution]{
\includegraphics[width=4cm]{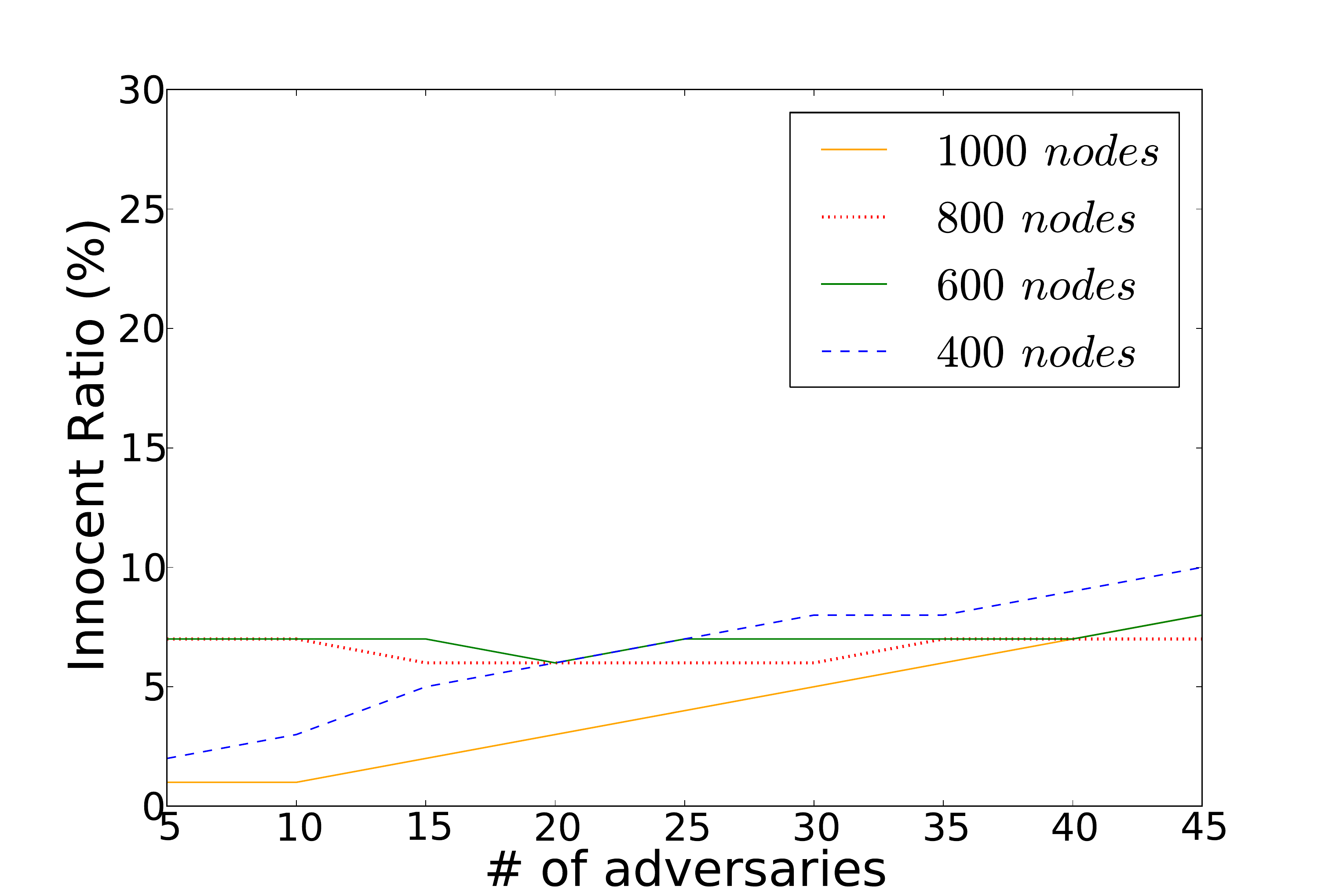}
\label{sf:in_gau}
}
\subfigure[Catch ratio of Gaussian distribution]{
\includegraphics[width=4cm]{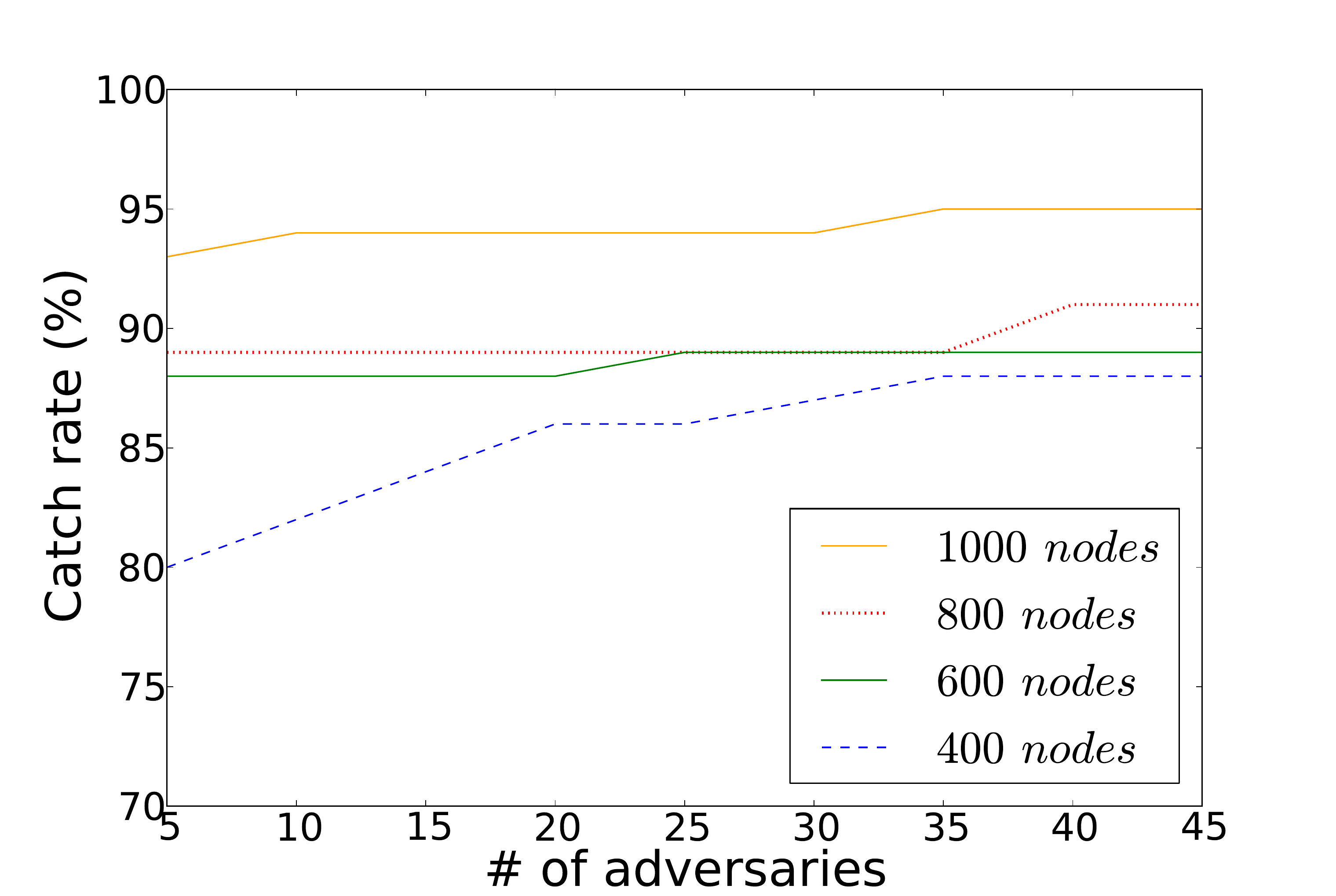}
\label{sf:cr_gau}
}
\caption{Innocent ratio and Byzantine catch ratio for two different distribution pattern of adversaries}
\label{result_origin}
\end{figure}
\begin{figure}[!t]
\centering
\subfigure[Innocent ratio of uniform distribution]{
\includegraphics[width=4cm]{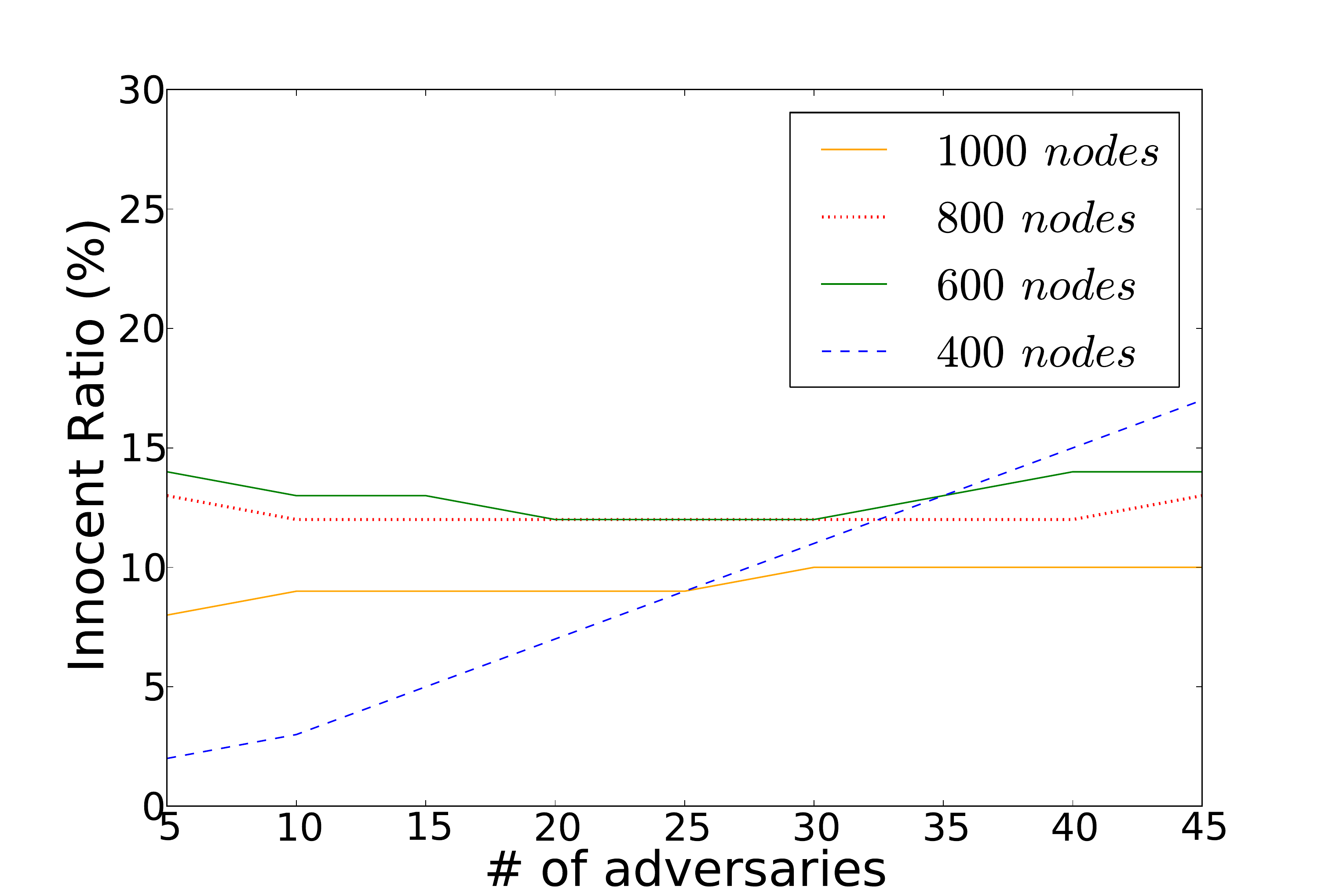}
\label{sf:in_uni_fs}
}
\subfigure[Catch ratio of uniform distribution]{
\includegraphics[width=4cm]{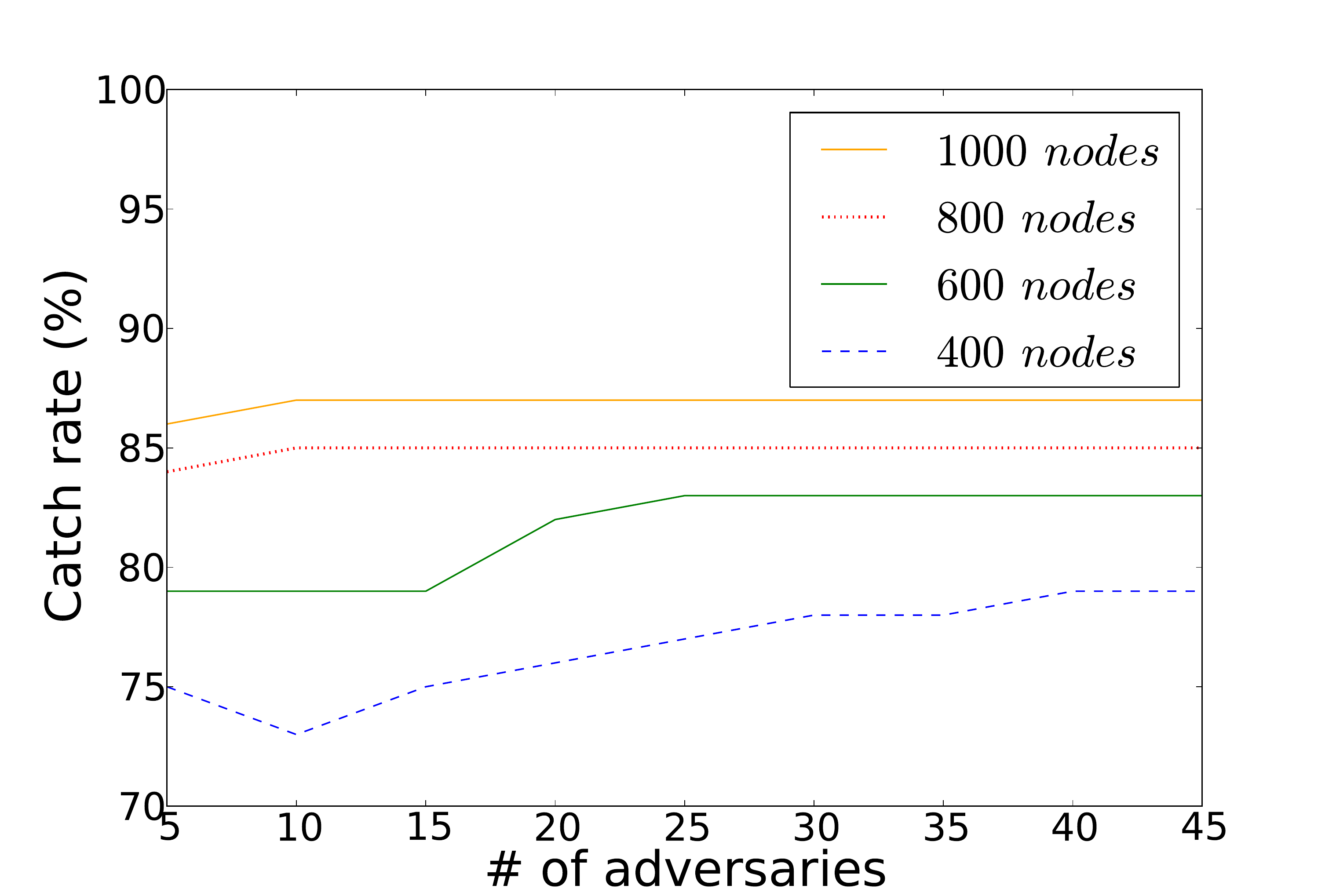}
\label{sf:cr_uni_fs}
}
\subfigure[Innocent ratio of Gaussian distribution]{
\includegraphics[width=4cm]{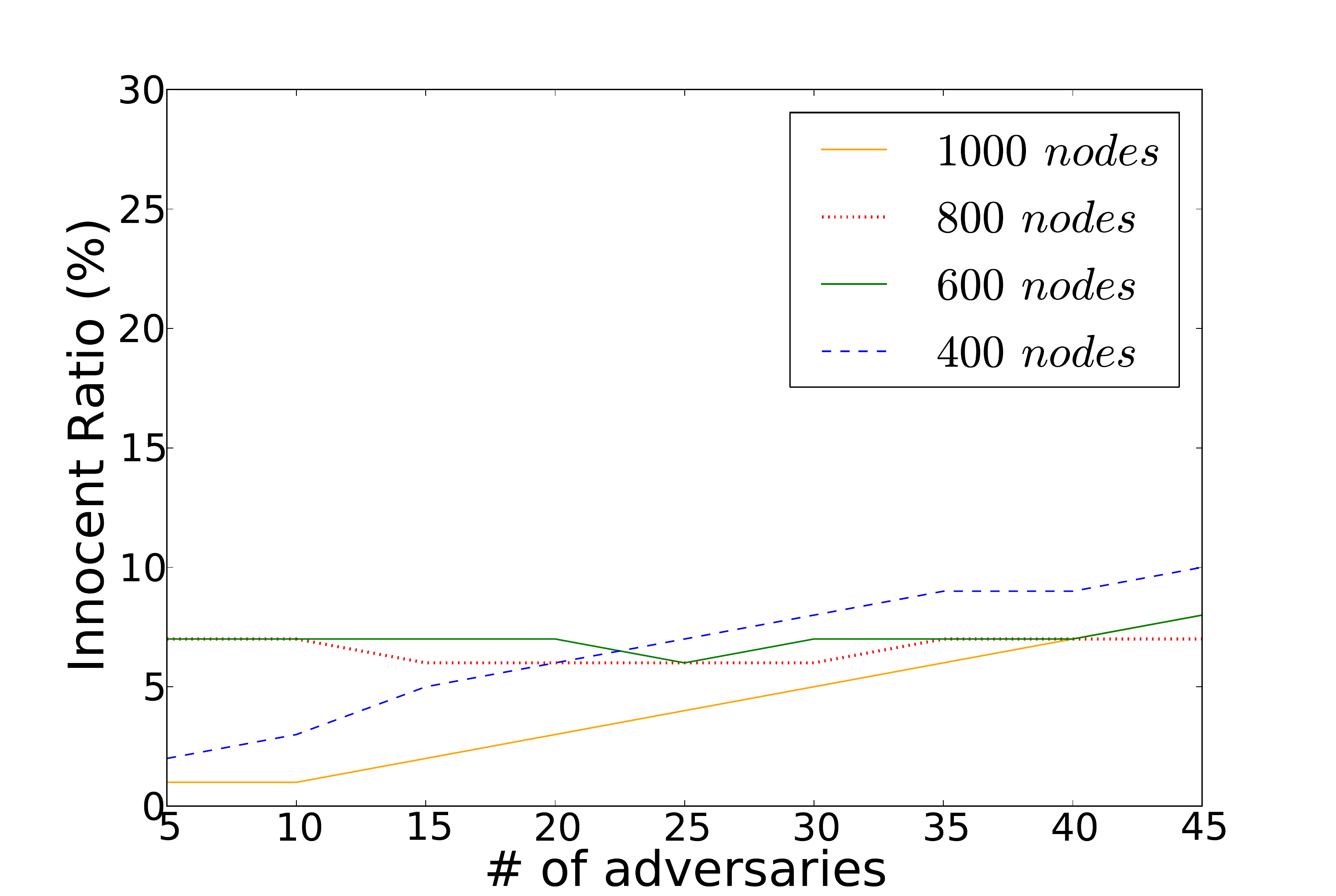}
\label{sf:in_gau_fs}
}
\subfigure[Catch ratio of Gaussian distribution]{
\includegraphics[width=4cm]{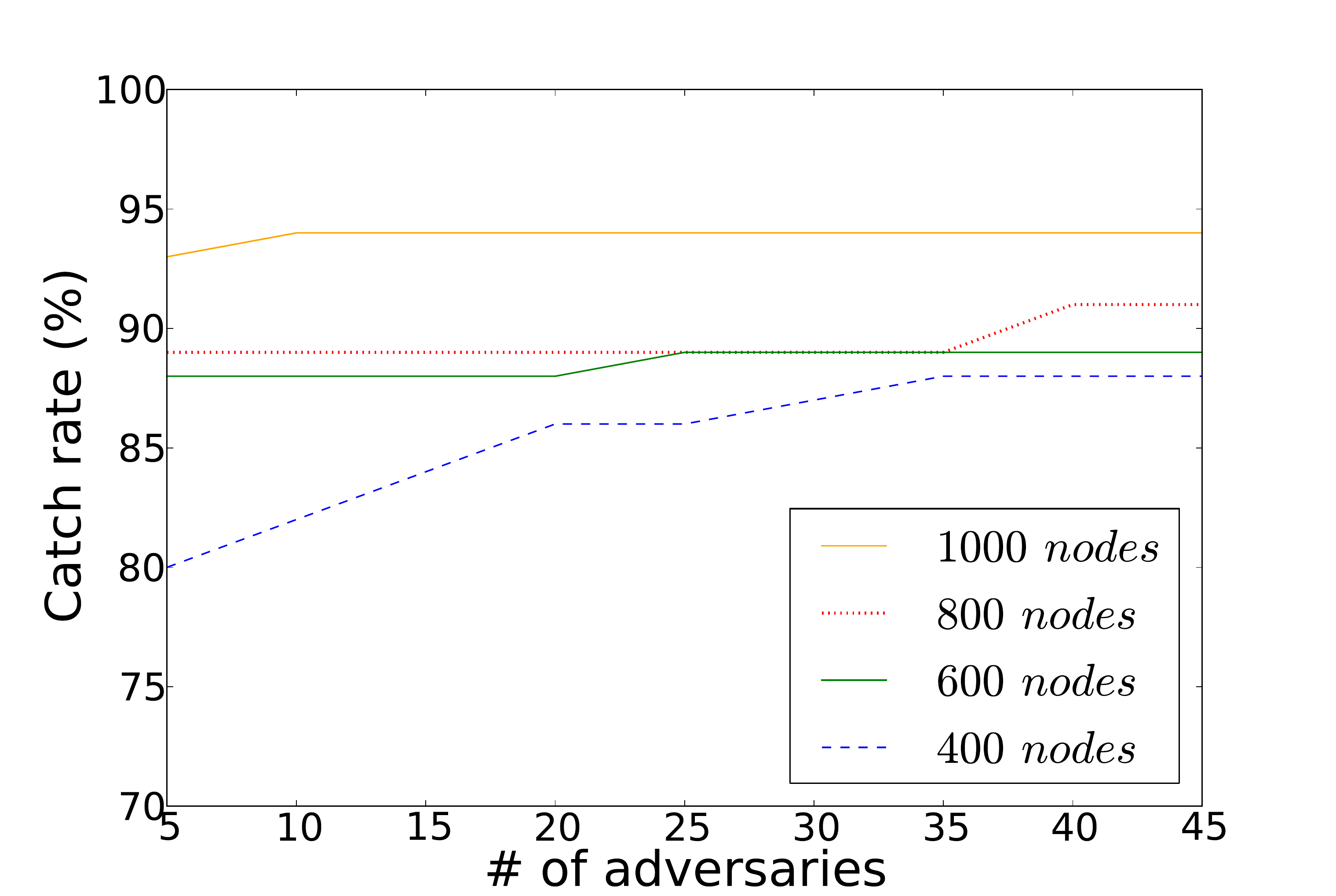}
\label{sf:cr_gau_fs}
}
\caption{Innocent ratio and Byzantine catch ratio with shift scheme}
\label{result_shift}
\end{figure}

\section{Analysis}
The shift scheme aims to further divide the least monitoring areas into smaller areas so that we can decrease the number of innocent nodes. With it, the final results of marked areas in each run of algorithm will be different. The overlapped marked areas are smaller than the least monitoring areas and contain less innocent nodes. Considering the case that overlapped areas divide a least monitoring area $A$ into four smaller areas, $A_1$, $A_2$, $A_3$ and $A_4$. The expectation number of innocent nodes will reach a minimum value while $A_1 = A_2 = A_3 = A_4$. We now prove our claim.
\begin{figure}[!t]
\centering
\subfigure[Innocent ratio of uniform distribution]{
\includegraphics[width=4cm]{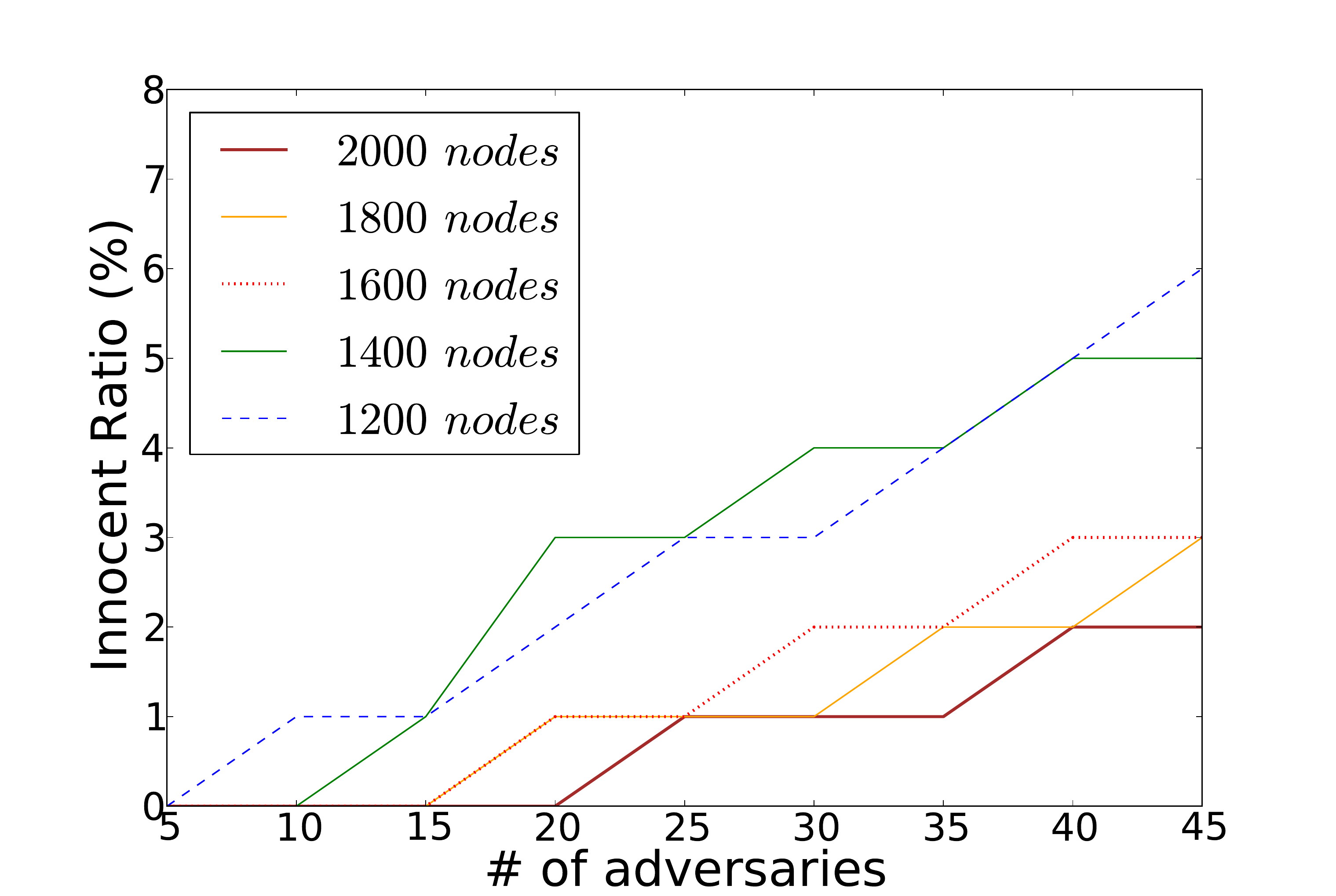}
\label{sf:in_uni_s125}
}
\subfigure[Catch ratio of uniform distribution]{
\includegraphics[width=4cm]{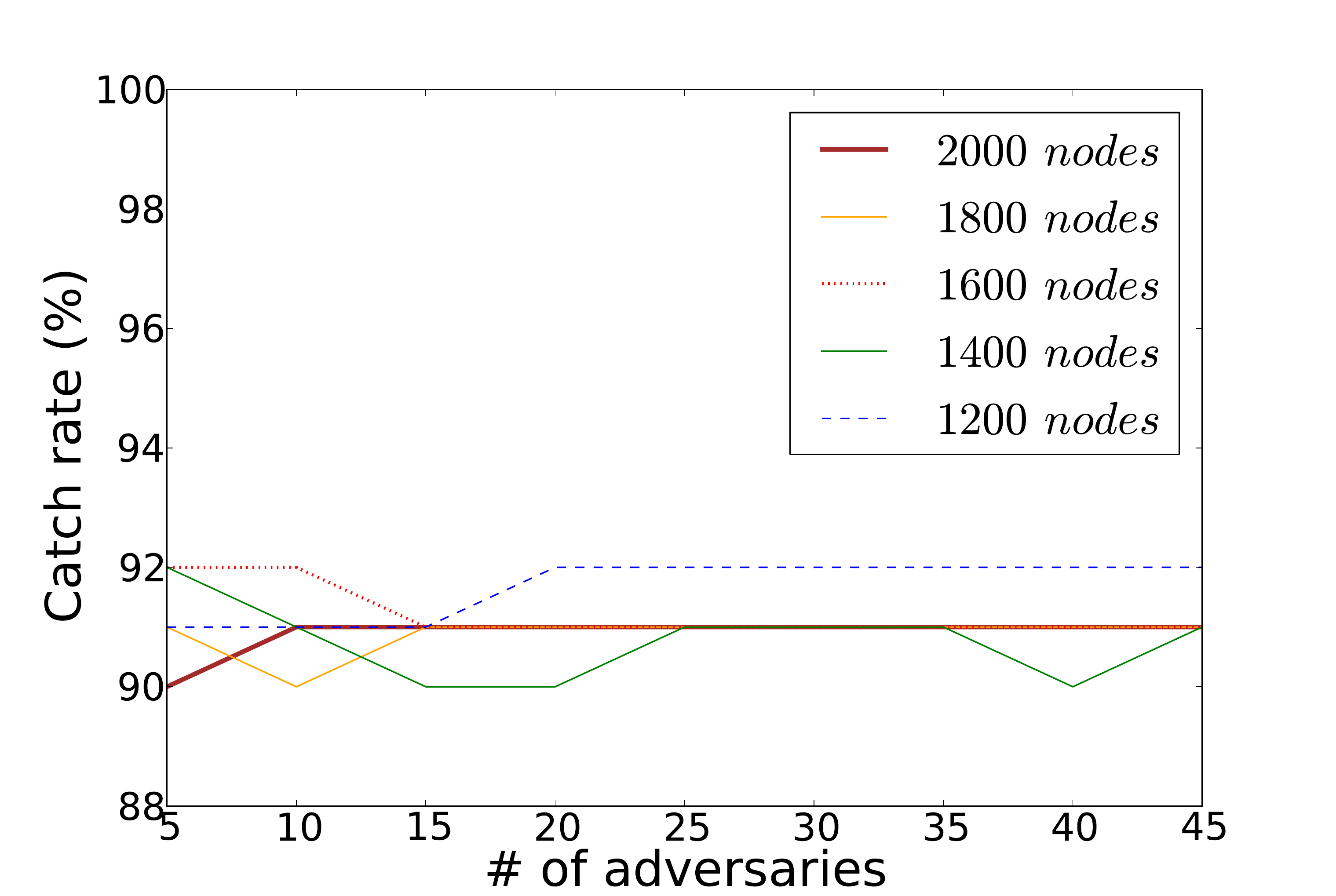}
\label{sf:cr_uni_s125}
}
\subfigure[Innocent ratio of Gaussian distribution]{
\includegraphics[width=4cm]{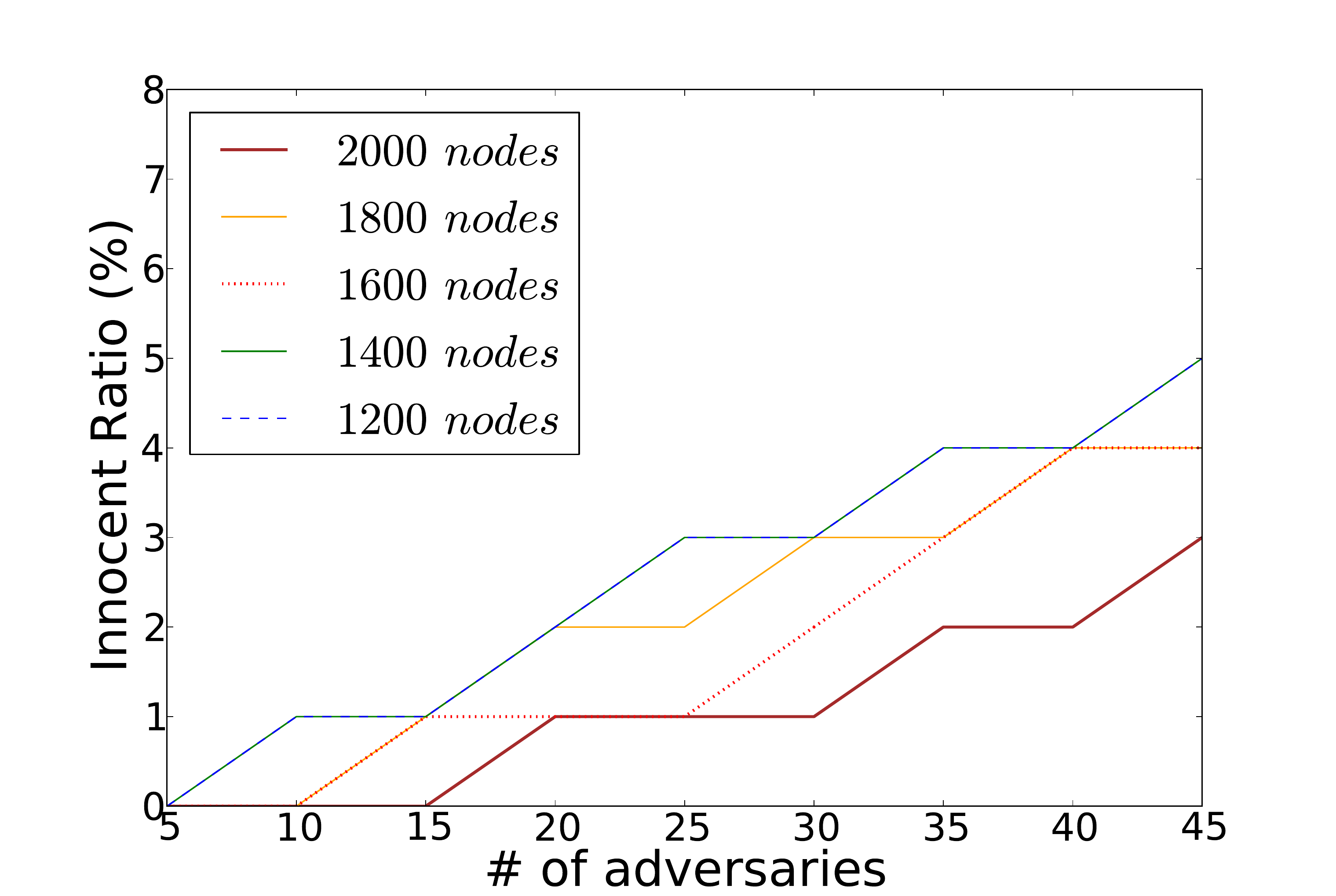}
\label{sf:in_gau_s125}
}
\subfigure[Catch ratio of Gaussian distribution]{
\includegraphics[width=4cm]{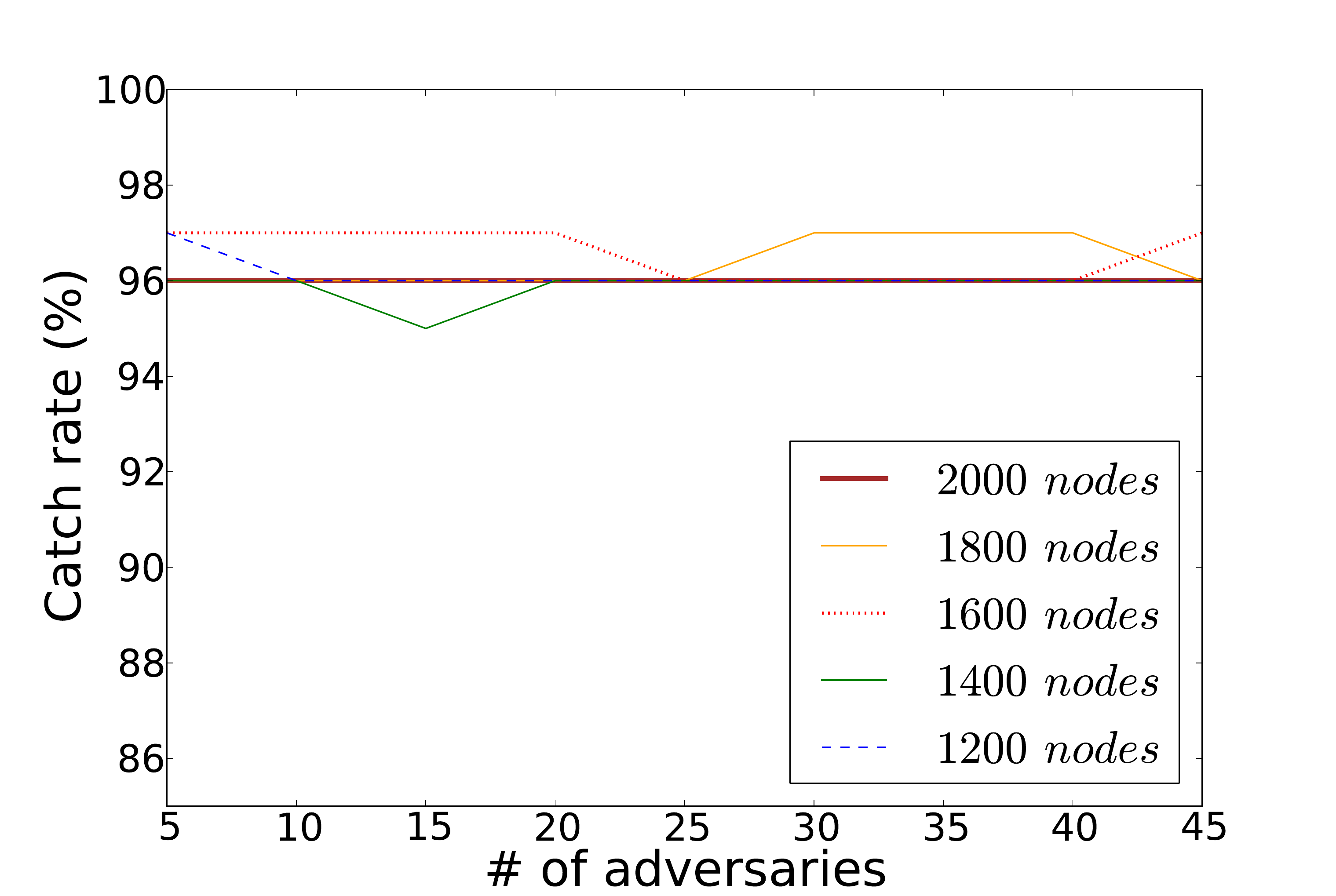}
\label{sf:cr_gau_s125}
}
\caption{Results for more nodes}
\label{result_more}
\end{figure}

\begin{claim}
The expectation value of number of innocent nodes will reach a minimum when the least monitoring area $A$ is divided into four equal areas.
\end{claim}

\begin{proof}
Assume that the area $A$ is of size $1$ and divided into four areas, $A_1$, $A_2$, $A_3$ and $A_4$, with the area size of $a_1$, $a_2$, $a_3$ and $a_4$. We have $a_1+a_2+a_3+a_4=1$ and $a_1$, $a_2$, $a_3$, $a_4 > 0$. The least monitoring area $A$ has $\mu$ nodes totally and $k$ of the $\mu$ nodes are adversary nodes. Clearly $k < \mu$\\
The expectation number of innocent nodes is 
\begin{subequations}
\begin{align*}
E(k)= &[1-(1-a_1 )^k]a_1 \mu +[1-(1-a_2 )^k]a_2 \mu + \\
      &[1-(1-a_3 )^k]a_3 \mu +[1-(1-a_4 )^k]a_4 \mu \\
= & (a_1 +a_2 +a_3 +a_4)\mu - \\
  &[a_1 (1-a_1 )^k +a_2 (1-a_2 )^k +a_3 (1-a_3 )^k +a_4 (1-a_4 )^k ]\mu \\
= & \mu-[a_1 (1-a_1 )^k +a_2 (1-a_2 )^k +a_3 (1-a_3 )^k +a_4 (1-a_4 )^k ]\mu.
\end{align*}
\end{subequations}
We want to have $E(k)\geq$ some constant $c$, so the problem becomes
\begin{subequations}
\begin{align*}
\mbox{maximize  }&~\mathbf{x_1 (1-x_1 )^k +x_2 (1-x_2 )^k +x_3 (1-x_3 )^k +x_4 (1-x_4 )^k} \\
\mbox{subject to  }&~\mathbf{x_1 +x_2 +x_3 +x_4 = 1}.
\end{align*}
\end{subequations}
We denote $f(\mathbf{x})=x_1 (1-x_1 )^k +x_2 (1-x_2 )^k +x_3 (1-x_3 )^k +x_4 (1-x_4 )^k$ and $h(\mathbf{x})=x_1+x_2+x_3+x_4-1$. By the Lagrange condition, we have
\begin{subequations}
\begin{align*}
(1-x_1)^k -&kx_1(1-x_1)^{k-1} + \lambda = 0 \\
(1-x_2)^k -&kx_1(1-x_2)^{k-1} + \lambda = 0 \\
(1-x_3)^k -&kx_1(1-x_3)^{k-1} + \lambda = 0 \\
(1-x_4)^k -&kx_1(1-x_4)^{k-1} + \lambda = 0 \\
x_1 +&x_2+x_3 +x_4 = 1.
\end{align*}
\end{subequations}
Obviously, the solution to theses equations is 
\[
x_1=x_2=x_3=x_4=\dfrac{1}{4}~~and~~\lambda=(\dfrac{k}{4}-\dfrac{3}{4})(\dfrac{3}{4})^{k-1}
\]
Thus $\mathbf{x^{\ast}}=[\frac{1}{4},\frac{1}{4},\frac{1}{4},\frac{1}{4}]^{\top}.$

Now we need to resort to the second-order sufficient conditions to determine if the problem reaches a maximum or minimum at $x_1=x_2=x_3=x_4=\dfrac{1}{4}$. Let $l(\mathbf{x},\mathbf{\lambda})=f(\mathbf{x})+\lambda^{\top}h(\mathbf{x})$ and $\mathbf{L(x,\lambda)}$ be the Hessian matrix of $l(\mathbf{x,\lambda})$. We can find the matrix 
\begin{subequations}
\begin{align*}
\mathbf{L(x^{\ast},\lambda)}~=&~ \mathbf{F(x^{\ast})}+\lambda \mathbf{H(x^{\ast})} \\~=&~
\begin{bmatrix}                
  \mathbf{g}(k) & 0 & 0 & 0 \\
  0 & \mathbf{g}(k) & 0 & 0 \\
  0 & 0 & \mathbf{g}(k) & 0 \\
  0 & 0 & 0 & \mathbf{g}(k) 
\end{bmatrix}, \\
\end{align*}
\end{subequations}
where $\mathbf{g}(k)~=~(\dfrac{3}{4})^{k-2}(\dfrac{k-7}{4})$. On the tangent space $M=\left\{\mathbf{y}\mid y_1+y_2+y3+y_4=0\right\}$, we note that 
\begin{subequations}
\begin{align*}
\mathbf{y^{\top}Ly}~=&~y_1^2(\dfrac{3}{4})^{k-2}(\dfrac{k-7}{4}) +y_2^2(\dfrac{3}{4})^{k-2}(\dfrac{k-7}{4}) + \\
&~y_3^2(\dfrac{3}{4})^{k-2}(\dfrac{k-7}{4}) +y_4^2(\dfrac{3}{4})^{k-2}(\dfrac{k-7}{4}) 
<~0,\\
&~\mbox{for $k<7$ and all $y\neq0$}.
\end{align*}
\end{subequations}
Thus $L$ is negative definite on $M$ when $k<7$ and $f$ reaches a maximum. In our algorithm, we set our $\mu=5$, and $k<\mu$ obviously. Therefore, we can always reach a minimum expectation value in our setup and it happens at $a_1=a_2=a_3=a_4=\dfrac{1}{4}$.

\end{proof}

\section{Conclusions and further work}
We have proposed a locating algorithm in appliance of RLNC to locate compromised Byzantine nodes in a network. Our algorithm can locate the areas in where adversary nodes locate with some normal nodes being mistaken as adversary nodes. To reduce the number of mistaken nodes, we use a shift scheme to eliminate the probability of being mistaken. The simulation results show that our algorithm performs well in Guassian distribution pattern for adversary nodes. In the worst case, uniform distribution pattern for adversary nodes, we still can locate most adversary nodes and reduce almost $10\%$ of mistaken ratio by shift scheme. We also gives discussion about the best policy for shift scheme. Fixing the shift range to the half length of the least monitoring area has the best performance.

Even though we do locate the areas where adversary nodes lie, but there still exist mistaken nodes. A second stage algorithm is required in order to precisely identify each adversary node. Sampling each node one by one in the most suspicious area or combining some special coding scheme with our algorithm may be a worthy researching direction.

\bibliographystyle{IEEEtran}
\bibliography{myref}

\end{document}